%% file: main.tex
\documentclass[submission,copyright,creativecommons,fleqn]{eptcs}
\providecommand{\event}{QAPL 2016} % Name of the event you are submitting to

\input{preamble}

\usepackage{amsthm}

\newtheorem{lemma}{Lemma}

\newtheorem{theorem}{Theorem}
\newtheorem{remark}{Remark}
\newtheorem{corollary}{Corollary}

\newtheorem{problem}{Problem}

\title{Limit Your Consumption!\\ Finding Bounds in Average-energy Games\thanks{Supported by the DFG project TriCS (ZI 1516/1-1), the ERC Advanced Grant LASSO and the EU FET projects SENSATION and CASSTING.}}

\author{Kim G.\ Larsen \qquad\qquad Simon Laursen
\institute{Department of Computer Science, Aalborg University\\ Aalborg, Denmark}
\email{kgl@cs.aau.dk \quad\qquad simlau@cs.aau.dk}
\and
Martin Zimmermann
\institute{Reactive Systems Group, Saarland University\\  Saarbr\"ucken, Germany}
\email{zimmermann@react.uni-saarland.de}
}
\def\titlerunning{Finding Bounds in Average-energy Games}
\def\authorrunning{K.G.\ Larsen, S.\ Laursen \& M.\ Zimmermann}
\begin{document}
\maketitle

\begin{abstract}
\input{abstract}
\end{abstract}

\input{intro}

\input{definitions}

\input{averageenergy}

\input{rechargegames}

\input{tradeoffs}

\input{conc}

\bibliographystyle{eptcs}
\bibliography{refs}

\end{document}

%% file: preamble.tex
\usepackage{graphicx} 
\usepackage[font=small]{subfig}
\usepackage{amsmath, amssymb}
\usepackage{todonotes}
\usepackage{wasysym}
\usepackage[T1]{fontenc}
\usepackage{hyperref}
\usepackage{wrapfig}

\makeatletter
\let\phi\varphi
\let\epsilon\varepsilon
\makeatother

\newcommand{\A}{\mathcal{A}}
\newcommand{\wei}{{w}}
\newcommand{\win}{\mathrm{Win}}
\newcommand{\play}{\rho}
\newcommand{\game}{\mathcal{G}}
\newcommand{\peak}{\mathrm{Peak}}
\newcommand{\prefs}{\mathrm{Prefs}}
\newcommand{\real}{\mathrm{Real}}
\newcommand{\last}{\mathrm{Last}}
\newcommand{\rep}{\mathrm{Rep}}
\newcommand{\res}{\mathrm{R}}
\newcommand{\capa}{{cap}}
\newcommand{\thres}{{t}}
\newcommand{\initial}{I}

\newcommand{\EL}{\mathrm{EL}}
\newcommand{\EGL}{\mathsf{Energy_{L}}}
\newcommand{\EGLU}{\mathsf{Energy_{LU}}}
\renewcommand{\AE}{\mathsf{AvgEnergy}}
\newcommand{\AEL}{\mathsf{AvgEnergy_L}}
\newcommand{\AELU}{\mathsf{AvgEnergy_{LU}}}
\newcommand{\REC}{\mathsf{Recharge}}
\newcommand{\AREC}{\mathsf{AvgRecharge}}
\newcommand{\PARITY}[0]{\mathsf{Parity}}
\newcommand{\MP}[0]{\mathsf{MeanPayoff}}
\newcommand{\CNTDWN}[0]{\mathsf{Countdown}}

\newcommand{\mem}{\mathcal{M}}

\newcommand{\update}{\mathrm{Upd}}
\newcommand{\nxt}{\mathrm{Nxt}}

\def\EXPTIME{\textsc{ExpTime}}
\def\TWOEXPTIME{\textsc{2ExpTime}}

{\bfseries}{\itshape}

\usepackage{tikz}
\usetikzlibrary{automata, positioning, decorations.pathmorphing}
\usetikzlibrary{arrows,shapes,backgrounds}
\usetikzlibrary{fit,backgrounds,calc}

\tikzset{box/.style={rectangle, draw, text centered, minimum height=6.5mm,  minimum width = 6.5mm}}%
\tikzset{circ/.style={circle, draw, text centered, minimum height=7mm}}%

%% file: abstract.tex
Energy games are infinite two-player games played in weight\-ed arenas with quantitative objectives that restrict the consumption of a resource modeled by the weights, e.g., a battery that is charged and drained. Typically, upper and/or lower bounds on the battery capacity are part of the problem description. Here, we consider the problem of determining upper bounds on the average accumulated energy or on the capacity while satisfying a given lower bound, i.e., we do not determine whether a given bound is sufficient to meet the specification, but if there exists a sufficient bound to meet it.

In the classical setting with positive and negative weights, we show that the problem of determining the existence of a sufficient bound on the long-run average accumulated energy can be solved in doubly-exponential time. Then, we consider recharge games: here, all weights are negative, but there are recharge edges that recharge the energy to some fixed capacity. We show that bounding the long-run average energy in such games is complete for exponential time. Then, we consider the existential version of the problem, which turns out to be solvable in polynomial time: here, we ask whether there is a recharge capacity that allows the system player to win the game. 

We conclude by studying tradeoffs between the memory needed to implement strategies and the bounds they realize. We give an example showing that memory can be traded for bounds and vice versa. Also, we show that increasing the capacity allows to lower the average accumulated energy. 

%% file: intro.tex
%!TEX root = main.tex

\section{Introduction} % (fold)
\label{sec:introduction}

Quantitative games provide a natural framework for synthesizing controllers
with resource restrictions and for performance requirements for reactive
systems with an uncontrollable environment. In a traditional
two-player graph game of infinite duration~(see~\cite{thomas2002automata}), two players, Player~$0$ (who represents the system to be synthesized) and Player~$1$ (representing the antagonistic environment), construct an infinite path by moving a pebble through a graph, which describes the interaction between the system and its environment. The objective, which encodes the controller's specification, determines the winner of such a play. Quantitative games extend classical ones by having weights on edges for modeling costs, consumption or rewards, and a quantitative objective to encode the specification in terms of the weights.

\begin{wrapfigure}{r}{.33\textwidth}
\vspace{-15pt}
\centering
\begin{tikzpicture}[auto,node distance=2 cm, yscale = 1,
		transform shape, thick]

     \tikzstyle{p0}=[draw,circle,text centered,minimum size=7mm,text width=4mm]
     \tikzstyle{p1}=[draw,rectangle,text centered,minimum size=7mm,text width=4mm]

\node[p1] 	(a)	at (0,1.8)				 	{$v_2$};
\node[p0]	(b)	at (1.5,0)	{$v_1$};
\node[p0]	(c)	at (-1.5,0)	{$v_0$};

\path[->,>=latex]
	(-2.3,0) edge [] node [] {} (c)
	(a) edge [bend left] 	node[below]	 	{-3\,\,\,\,\,\,}	(b)
	(a) edge [bend right] 	node[below]		{\,\,\,\,2}	(c)
	(b) edge [bend right = 15] 	node[above]	 	{0}	(c)
	(c) edge [bend right = 15] 	node[below]	 	{-1}	(b)
	(b) edge [bend right = 85] 	node[above]	 	{\,\,\,\,-1}	(a)
	(c) edge [bend left = 85] 	node[above]			{3\,\,\,\,}	(a)
; 
	\end{tikzpicture} 
	\vspace{-17pt}
\end{wrapfigure}
Consider the game depicted to the right: we interpret negative weights as energy consumption and correspondingly positive weights as recharges. Then, Player~$0$ (who moves the pebble at the circled vertices) can always maintain an energy level (the sum of the weights seen along a play prefix starting with energy~$0$) between zero and five using the following strategy: when at vertex~$v_0$ with non-zero energy level go to vertex~$v_1$, otherwise go to vertex~$v_2$ in order to satisfy the lower bound. At vertex~$v_1$ she moves to $v_0$ if the energy level is zero, otherwise to $v_2$. It is straightforward to verify that the strategy has the desired property when starting at the initial vertex~$v_0$ with initial energy~$0$. However, this strategy requires memory to implement, as its choices depend on the current energy level.

Quantitative games~\cite{BCHJ09,emsoft2003-CAHS,Ran13} and objectives such as mean-payoff~\cite{BCDGR11,VelnerC0HRR15,Zwick96}, energy \cite{bouyer2008,DBLP:journals/tcs/ChatterjeeD12,Juhl13}, and their combination \cite{Chatterjee2013} have attracted considerable attention recently. The focus has been on establishing the computational complexity of deciding whether Player~$0$ wins the game and on memory requirements. In mean-payoff games, Player~$0$'s goal is to optimize the long-run average gain per edge taken, whereas in energy games the goal is to keep the accumulated energy within given bounds. Recently, the average-energy objective was introduced~\cite{AvENgergy15} to capture the specification in an industrial case study~\cite{CJLRR09}. In this study, the authors synthesize a controller to operate an oil pump using timed games and \textsc{Uppaal} TiGA. The controller has to keep the amount of oil in an accumulator within given bounds while minimizing the average amount of oil in the accumulator in the long run. A discrete version of this problem is exactly an average-energy game, where the goal for Player~$0$ is to optimize the long-run average accumulated energy during a play while keeping the accumulated energy within given bounds.

Recall the introductory
example above. The strategy for Player~$0$ described there realizes the
long-run average~$4$: the consistent play~$v_0 (v_2 v_0 v_1)^\omega$ with energy levels $0, (3,5,4)^\omega$ has average~$4$, obtained by dividing the sum of the levels in the period by the
length of the period. Every other consistent play has a smaller or equal average.

The computational complexity of these quantitative objectives are typically studied with respect to given bounds on the energy level or given thresholds on the
mean-payoff or on the average accumulated energy. In this work, we consider the
variants where the bounds and thresholds are existentially quantified instead
of given as part of the input, i.e., we ask if there exist bounds and
thresholds such that Player~$0$ has a winning strategy. This question is
natural for models with bounds and thresholds as it desirable to
know if a given model is realizable for some bounds. In a second step, one would then determine the minimal bounds for which Player~0 is able to win.

In particular, we study existential questions on two different game models,
average-energy games and average-bounded recharge games. Average-energy games
are defined as in~\cite{AvENgergy15} with both positive and negative weights on
edges whereas in average-bounded recharge games all weights are negative, but
there are designated recharge-edges that recharge the energy to some fixed
capacity.

\medskip
\noindent\textbf{Our contribution.} % (fold)
For average-energy games, we show that the problem of deciding whether there
exists a threshold to which Player~0 can bound the long-run average accumulated
energy while keeping the accumulated energy non-negative can be solved in
doubly-exponential time. To this end, we show that the problem is equivalent
to determining whether the maximal energy level can be uniformly bounded by a
strategy. The latter problem is known to be in~$\TWOEXPTIME$~\cite{Juhl13}. The challenging
part is to construct a strategy that uniformly bounds the energy from the
strategy that only bounds the long-run average accumulated energy, but might reach
arbitrarily high energy levels. But whenever the energy level increases above the given
threshold, it has to drop below it at some later point. Thus, we can always
play like in a situation where the peak between these two threshold crossings is as small as possible. This 
yields a new strategy that bounds the energy level. Our result is one step
towards solving the open problem of solving lower-bounded average-energy games with a
given threshold~\cite{AvENgergy15}.

For average-bounded recharge games, we show that given a bound on the long-run
average energy, deciding the winner is~\EXPTIME-complete. For the existential
versions of the problem, we show that it remains $\EXPTIME$-hard when the
recharge capacity is quantified and the average threshold is given. The problem
becomes solvable in polynomial time when only the recharge capacity is considered:
here, we ask whether there is a recharge capacity such that Player~$0$ wins the
game with respect to this capacity.

Finally, we study tradeoffs between the different bounds and the memory
requirements of winning strategies, and show that increasing the upper bound on
the maximal energy level allows to improve the average energy level and memory
can be traded for smaller upper bounds and vice versa.

% paragraph contribution (end)

\input{relatedwork}

% section introduction (end)

%% file: relatedwork.tex
%!TEX root = main.tex

\medskip
\noindent\textbf{Related work.} % (fold)
The average energy objective was first introduced in \cite{TV87} under the name total-reward but has until recently not undergone a systematic study. Independently, it was studied (under the name total-payoff) for Markov decision processes and stochastic games~\cite{BEGM15}, and \cite{AvENgergy15} presented a comprehensive investigation into the problem in the deterministic case. The latter also considered extensions where the average-energy objective is combined with bounds on the energy, which is the model we consider here.

Several other games with combined objectives have been introduced such as mean-payoff parity~\cite{ChatterjeeHJ05}, energy-parity \cite{DBLP:journals/tcs/ChatterjeeD12}, multi-dimensional energy \cite{FahrenbergJLS11}, multi-dimensional mean-payoff \cite{VelnerC0HRR15} and the combination of multi-dimensional energy, mean-payoff and parity   \cite{Chatterjee2013}. In \cite{DBLP:conf/cav/BrazdilCKN12}, consumption games are studied where edges only have negative weights, and some distinguished edges recharge the energy to a level determined by Player~$0$. This model is related to recharge games, but in recharge games the recharge capacity is given and we consider average-bounded objectives.
Existential questions in games have been studied before in the form of determining the emptiness of a set of bounds that allow Player~$0$ to win a quantitative game, e.g., for multi-dimensional energy games with upper bounds~\cite{Juhl13} and for games with objectives in parameterized generalizations of LTL~\cite{AlurETP01,FaymonvilleZ14,KupfermanPV09,Z15}.

% paragraph related_work (end)

%% file: definitions.tex
%\include{main.tex}
\section{Definitions} % (fold)
\label{sec:definitons}

	An \emph{arena}~$\A = (V, V_0, V_1, E, v_{\initial})$ consists of a finite
	directed graph~$(V,E)$ without terminal vertices, a partition~$V = V_0 \uplus V_1$ of the vertices, and an initial vertex~$v_{\initial} \in V$.
Vertices in $V_0$ are under Player~$0$'s control and are drawn as circles, whereas vertices in
$V_1$ are under Player~$1$'s control and drawn as rectangles. A play in $\A$ is an infinite path $\play = v_0 v_1 v_2 \cdots $ with $v_0 =
v_{\initial}$. 
A \emph{game} $\game = (\A, \win)$ consists of an arena~$\A$, and a set 
$\win \subseteq V^\omega$ of winning plays for Player~$0$, the \emph{objective} of $\game$. The objectives we consider are induced by weight functions, assigning integer weights to edges, which are encoded in binary. We say an algorithm runs in \emph{pseudo-polynomial time}, if it runs in polynomial time in the number of vertices and in the largest absolute weight. An algorithm runs in polynomial time, if it runs in polynomial time in the number of vertices and in the size of the encoding of the largest absolute weight.

A \emph{strategy} for Player $i \in \{0,1\}$ is a mapping
$\sigma_i \colon V^*V_i \to V$ such that $(v,\sigma_i(w v)) \in E$ for all $w v \in
V^*V_i$. A play $v_0 v_1 v_2 \cdots $ is \emph{consistent} with a strategy~$\sigma_i$ for Player $i$ if $v_{n+1} = \sigma_i(v_0 v_1 \cdots v_n)$ for every
$n$ with $v_n \in V_i$. 
A strategy $\sigma_0$ for Player~$0$ is winning for the game 
$\game = (\A,\win)$ if every play that is consistent with $\sigma_0$ is 
in $\win$. We say that Player~$0$ wins $\game$ if she has a winning strategy for $\game$. We define $\prefs(\sigma)$ to denote the set of finite play prefixes that are consistent with $\sigma$.
We denote the last vertex of a non-empty word~$w$ by~$\last(w)$.

A \textit{memory structure}~$\mem = (M, m_{\initial}, \update)$ for an arena $(V, V_0, V_1,
E, v_{\initial})$ consists of a finite set~$M$ of memory states, an initial memory state~$m_{\initial} \in M$, and an update function~$\update
\colon M \times E \rightarrow M$. The update function can be extended to
$\update^+ \colon V^+ \rightarrow M$ in the usual way: $\update^+(v_0) =
m_{\initial}$ and $\update^+ (v_0 \cdots v_n v_{n+1}) =
\update(\update^+(v_0 \cdots v_n), (v_n,v_{n+1}))$. A next-move function
(for Player~$i$) $\nxt \colon V_i \times M \rightarrow V$ has to satisfy $(v,
\nxt(v, m)) \in E$ for all $v \in V_i$ and all $m \in M$. It induces a
strategy~$\sigma$ for Player~$i$ via
$\sigma(v_0\cdots v_n) = \nxt(v_n, \update^+(v_0 \cdots v_n))$.
A strategy is called \textit{finite-state} (\emph{positional}) if it can be implemented by a memory
structure (with a single state). Intuitively, the next move of a positional strategy only depends on the last vertex of the play prefix.
An arena $\A = (V, V_0, V_1, E, v_{\initial})$ and a memory structure $\mem = (M, m_{\initial},
\update)$ for $\A$ induce the expanded arena $\A\times\mem = (V \times
M, V_0 \times M, V_1 \times M, E' , (v_{\initial}, m_{\initial}))$ where $((v,m), (v',m')) \in E'$ if and
only if $(v,v') \in E$ and $\update(m, (v,v') ) = m'$. Each play $v_0 v_1 v_2 \cdots$ in
$\A$ has a unique extended play $(v_0, m_0) (v_1, m_1)
(v_2, m_2) \cdots$ in $\A \times \mem$ defined by $m_0 = m_{\initial}$  and $m_{n+1} = \update(m_n, (v_n,v_{n+1}))$, i.e., $m_n = \update^+(v_0 \cdots v_n)$.
A game $\game = (\A, \win)$ is \textit{reducible} to $\game' = (\A', \win')$
via $\mem$, written $\game \le_{ \mem } \game'$, if $\A' = \A \times
\mem$ and every play $\play$ in $\game$ is won by the player who wins the
extended play $\play'$ in $\game'$, i.e., $\play \in \win $ if, and only if, $ \play' \in \win'$.

\begin{lemma} 
\label{lemma_reductionlemma} 
If $\game \le_{\mem } \game'$ and Player~$i$ has a positional winning strategy for $\game'$, then she has a finite-state winning
strategy for $\game$ which is implemented by $\mem$.
\end{lemma}

%% file: averageenergy.tex
%!TEX root = main.tex
\section{Finding Bounds in Average-energy Games} % (fold)
\label{sec:existential_average_energy_games}
	
	In this section, we study average-energy games with existentially quantified bounds on the average accumulated energy: our main theorem shows that these games are solvable in doubly-exponential time. 

	A weight
	function for  an arena $(V, V_0, V_1, E, v_{\initial})$ is a function~$ \wei \colon E \to \mathbb{Z} $ mapping every edge to an integer weight. The energy level of a play prefix is the accumulated weight of its edges, i.e., 
			$\EL(v_0 \cdots v_n) = \sum\nolimits_{i=0}^{n-1}
		\wei(v_i,v_{i+1})$.
	We consider several objectives obtained by specifying upper and lower bounds on the energy level and on the long-run average accumulated energy.
\begin{itemize}
	\item The lower-bounded energy objective requires Player~$0$  to keep the energy level non-negative: 
		\[
	\EGL(\wei) = \{ v_0 v_1 v_2 \cdots \in V^\omega \mid 
	\forall n.\, 0 \le \EL(v_0 \cdots v_n)  \}
	\]
	\item 
The lower- and upper-bounded energy objective requires Player~$0$  to keep the energy level always between $0$ and some given upper bound~$\capa$, the so-called capacity:
\[
	\EGLU(\wei, \capa) = \{ v_0 v_1 v_2 \cdots \in V^\omega \mid 
	\forall n.\, 0 \le \EL(v_0 \cdots v_n)  \le \capa \}
	\]
	\item The average-energy objective requires Player~$0$ to keep the long-run average of the accumulated energy below a given threshold~$\thres$: \[
	\AE(\wei, \thres) = \{  v_0 v_1 v_2 \cdots \in V^\omega \mid 	
	\limsup_{n \to \infty} \frac{1}{n} \sum\nolimits_{i = 0}^{n-1} 
	\EL(v_0 \cdots v_i)  \le \thres \}
	\]
	\item Also, we consider conjunctions of objectives, i.e., the lower-bounded average-energy objective
	\[
	\AEL(\wei, \thres) = \EGL(\wei) \cap \AE(\wei, \thres) 
	\]
	and the	lower- and upper-bounded average-energy objective
	\[
	\AELU(\wei, \capa, \thres) = \EGLU(\wei, \capa) \cap \AE(\thres).	
	\]
\end{itemize}
Note that we always assume the initial energy level to be zero. This is not a restriction, as one can always add a fresh initial vertex with an edge to the old initial vertex that is labeled by the desired initial energy level. Similarly, one can reduce arbitrary non-zero lower bounds to the case of the lower bound being zero, which is the one we consider here. 

Decidability of determining the winner of a game with lower-bounded av\-er\-age-energy objective with a given threshold~$\thres$ is an open problem~\cite{AvENgergy15}. To take a step towards solving this problem, we consider the existential variant of the problem, i.e., we ask whether there exists some threshold~$\thres$ such that Player~$0$ wins the game with objective~$\AEL(\wei, \thres)$:

\begin{problem} \label{prob:threshold}
	Existence of a threshold in a lower-bounded average-energy game.  \\ 
	\textbf{Input}: Arena $\A = (V, V_0, V_1, E, v_{\initial})$ and 
	$ \wei \colon E \to \mathbb{Z} $ \\
	\textbf{Question}: Exists a threshold~$t \in \mathbb{N}$ s.t.\ Player~$0$ wins $(\A, \AEL(\wei, \thres))$?
\end{problem}

We show that this problem is reducible to asking for the existence of an upper bound on the capacity~$\capa$. Note that such an upper bound also bounds the average accumulated energy. However, the converse is non-trivial as the average can be bounded while the energy level is unbounded. Formally, we consider the following problem: 

\begin{problem} \label{prob:upper}
	Existence of an upper bound in a lower- and upper-bounded energy game.  \\ 
	\textbf{Input}: Arena $\A = (V, V_0, V_1, E, v_{\initial})$ and 
	$ \wei \colon E \to \mathbb{Z} $ \\
	\textbf{Question}: Exists a capacity~$\capa \in \mathbb{N}$ s.t.\ Player~$0$ wins $(\A, \EGLU(\wei, \capa))$?
\end{problem}

The main theorem of this section shows that the existence of a threshold in a lower-bounded average-energy game can be checked in doubly-exponential~time. Our choice of encoding the weights influences the complexity of the problem: if the weights are encoded in unary, then the complexity drops to $\EXPTIME$. Furthermore, the problem is trivially at least as hard as solving mean-payoff games.

\begin{theorem}
\label{theorem_exavenergy}
	The threshold problem for lower-bounded average-energy games is in $\TWOEXPTIME$.
\end{theorem}

To prove this theorem, it suffices to show that Problem~\ref{prob:threshold} and Problem~\ref{prob:upper} are equivalent, as the latter problem was shown to be in $\TWOEXPTIME$~\cite{Juhl13}. 

\begin{lemma} 
	Let $\A$ be an arena and let $\wei$ be a  weight function for $\A$. Player~$0$ wins $(\A, \AEL(\wei, \thres))$ for some $t \in \mathbb{N}$ if, and only if,
Player~$0$ wins $(\A, \EGLU(\wei, \capa))$ for some $\capa \in
\mathbb{N}$.
\end{lemma}

\begin{proof}
	It is clear that a winning strategy $\sigma$ for 
	$(\A,\EGLU(\wei, \capa))$ for some $\capa \in \mathbb{N}$ is a 
	winning strategy for $(\A, \AEL(\wei, \capa))$,
	as if the energy level is always below some $\capa$, then the average energy is also bounded by $\capa$. 
	
For the other direction, assume that $\sigma$ is a winning strategy for
	Player~$0$ in $(\A, \AEL(\wei, \thres))$ for some 
	$\thres \in \mathbb{N}$. Now, we want to construct a strategy~$\sigma'$
	that is winning for Player~$0$ in $(\A,\EGLU(\wei, \capa))$ for some 
	$\capa \in \mathbb{N}$. Note that $\sigma$ might bound the average to some value while the energy level might be unbounded. But whenever the energy level increases above~$t$, it has to drop below $t$ at some point. We use this property to construct a strategy~$\sigma'$ that bounds the energy level.

	 First, we need to introduce some notation. Fix a play prefix $w \in \prefs(\sigma)$ 
	with $\EL(w) > \thres$ and define
	\[
		\peak(w) = \sup \{\EL (wx)\mid wx \in \prefs(\sigma) \text{ and }
		\EL(wx') > t \text{ for all } x' \sqsubseteq x\},
	\]
	i.e., $\peak(w)$ is the supremum of the energy levels of prolongations of $w$ that are consistent with $\sigma$ and have not yet had an energy level below~$t$. Applying K\"onig's Lemma~\cite{Konig27} and the fact that $\sigma$ is a winning strategy implies that the peak is always bounded.  

	\begin{remark}
		We have $\peak(w) \in \mathbb{N}$ for every $w \in \prefs(\sigma)$.
	\end{remark}
	
	For an energy level $c \in \mathbb{N}$ and a vertex $v \in V$ we define the 
	set of possible ways to end up in vertex $v$ with the energy level $c$
	playing consistently with $\sigma$ as
	\[
		\real(v,c) = \{ w \in \prefs(\sigma) \mid\last(w) = v  \text{ and } \EL(w) = c\}.
	\] 
	For every combination~$(v,c)$ with $c >t$, we pick a representative from $\real(v,c)$ that minimizes the peak height among all such realizations, i.e., we define
	$
		\rep(v,c)$ to be an element~$w$ from $\real(v,c)$ with minimal peak-value~$\peak(w)$ among the play prefixes in $\real(v,c)$. Note that $\rep(v,c)$ might be undefined, i.e., if there is no play prefix ending in $v$ with energy level $c$.
	
Intuitively, we construct a new strategy that mimics the behavior of $\sigma$ until the energy level increases above $\thres$. At this point, the history is replaced by the representative for the last vertex and the current energy level. Then, our new strategy mimics the behavior of $\sigma$ with this history until the threshold~$\thres$ is again crossed from below. Then, the next representative is picked. This strategy satisfies an upper bound, as only a finite number of representatives, each with a bounded peak-value, are considered when mimicking $\sigma$. 
To formalize this, we recursively define $h \colon V^+ \to \prefs(\sigma)$ via
	$h(v_{\initial}) = v_{\initial}$ and 
	\begin{align*}	
		h(wv) = \begin{cases}  \rep(v,\EL(h(w)v)) & \mbox{if }
		 \EL(h(w)) \le t  \text{ and }  \EL(h(w)v) > t \\ 
		h(w)v & \mbox{ otherwise }  \end{cases}
	\end{align*}
	for a play prefix 
	$wv \in V^+$ ending in a vertex $v$, i.e., $h(w)$ is the play prefix that simulates $w$. Now, we define the new strategy $\sigma'$ via	$\sigma'(w) = \sigma(h(w))$. The following remark implies that this is well-defined, although $\rep$ and therefore $h$ and $\sigma$ might be undefined for certain inputs. 
	
	\begin{remark}
	\label{rem_welldefinedness}
Let $w$ be consistent with $\sigma'$. Then, $h(w)$ is defined and consistent with $\sigma$, $\last(w) = \last(h(w))$, and $\EL(w) = \EL(h(w))$.
%, and 
%		\begin{enumerate}
%			\item if $\last(w) \in V_0$, then $\sigma'(w)$ and $h(w \sigma'(w))$ are well-defined, and
%			\item if $\last(w) \in V_1$, then $h(w v)$ is well-defined for every $v$ with $(\last(w),v) \in E$. 
%		\end{enumerate}
	\end{remark}

	Applying the remark inductively we conclude that $h(w)$ is defined for	every play prefix~$w$ that is consistent with $\sigma'$. This implies that $\sigma'(w)$ is well-defined for every such $w$ that ends in a vertex from $V_0$. Furthermore, this also implies that $\sigma'$ still satisfies the lower bound on the energy level. 

Thus, it remains to prove that an upper bound exists. Let $\play = v_0 v_1 v_2 \cdots$ be consistent with $\sigma'$ and let $n$ be such that $\EL(v_0 \cdots v_n) \le t$ and
	$\EL(v_0 \cdots v_n v_{n+1}) > t$. If there is no such $n$, then $\sigma$ bounds the energy level by $t$ and we are done. Furthermore, define $n'$ to be minimal with 
	$n' > n+1$ and $\EL(v_0 \cdots v_{n'}) \le t$ and
	$\EL(v_0 \cdots v_{n'} v_{n'+1}) > t$ (if no such $n'$ exists the reasoning is analogous). 
	As the energy level between the positions $n+1$ and $n'$ never crosses the
	threshold~$\thres$ from below, we are always in the second case of the definition of
	$h$. Thus, after the play prefix~$v_0 \cdots v_{n+1}$, the strategy~$\sigma'$  
	mimics the behavior of $\sigma$ after the prefix $h(v_0 \cdots v_{n+1}) =
	\rep(v_{n+1}, \EL(v_0 \cdots v_{n+1}))$. 
	Therefore, the energy level between these two positions is bounded by	$\peak(\rep(v_{n+1}, \EL(v_0 \cdots v_{n+1})))$. As we only take those
	representatives into account that have an energy level between $t+1$ and
	$t+W$, where $W$ is the largest positive weight in the image of $\wei$, the energy level of the
	play is bounded by the maximal peak of one of these representatives. Finally, this bound 
	is uniform for all plays that are consistent with $\sigma'$. Thus, $\sigma'$ is winning in the game~$(\A, \AEL(\wei, \capa))$ for some $\capa$.
\end{proof}

Note that we do not obtain any upper bounds on the energy level or on the long-run average energy realized by $\sigma'$, as they depend on properties of $\sigma$. One can even construct examples that show these values to be arbitrarily large by starting with a \textit{bad} winning strategy~$\sigma$ for the energy game. 

% section existential_average_energy_games (end)

% section reduction_to_existential_u (end)

%% file: rechargegames.tex
%!TEX root = main.tex

\section{Finding Bounds in Average-bounded Recharge Games} % (fold)
\label{sec:recharge_games}

In this section, we study a variation of energy games called recharge games (the name is inspired by recharge automata, first introduced in~\cite{EF13}). In such games, there are designated recharge edges that recharge the energy to some given capacity. All other edges have non-positive cost, i.e., they only decrease the energy level or leave it unchanged. This is reminiscent of so-called consumption games~\cite{DBLP:conf/cav/BrazdilCKN12}, where Player~$0$ picks the new energy level while traversing a recharge edge. There, one is interested in which initial energy levels allow Player~$0$ to win and to compute upper bounds on the recharge levels picked by Player~$0$. 

In this section, we go beyond just bounding the energy level by also considering bounds on the average accumulated energy, as we have done for average-energy games. However, the resulting games are intractable, as soon as the threshold on the average is part of the input. These results are presented in Subsection~\ref{sub:complexity}. To overcome the high complexity, in Subsection~\ref{sub:existence_of_cap} we consider the problem where the recharge capacity is existentially quantified: this problem is solvable in polynomial time by a reduction to three-color parity games. 

Here, we consider weight functions with only
non-positive weights and a special recharge action~$\res$, i.e., $\wei \colon E \to -\mathbb{N} \cup \{ \res \}$. The recharge action~$\res$ returns the energy level to some given upper bound capacity $\capa$. The
recharge energy level is the energy left since the last recharge action, which is defined as $\EL_{\capa}(v_0 \cdots v_n) = \capa + \EL(x) $, where $x$ is
the longest suffix of $v_0 \cdots v_n$ without an $\res$-edge, i.e., $\wei(v_j, v_{j+1}) \neq \res$ for all $(v_j, v_{j+1})$ in $x$, which implies that a play starts with energy level~$\capa$. We
define the objective of a recharge game as
\[
	\REC(\wei,\capa) = \{ v_0 v_1 v_2 \cdots \in V^\omega \mid 
	\forall n.\,   \EL_{\capa}(v_0 \cdots v_n) \ge 0 \}
\]
and the average-bounded version as
\vspace{-2pt}
\begin{align*}
	\AREC( \wei, \capa, \thres ) =  \{ & v_0 v_1 v_2 \cdots \in V^\omega \mid \limsup_{n \to \infty} \frac{1}{n} \sum_{i = 0}^{n-1} 
	\EL_{\capa}(v_0  \cdots v_i)  \le \thres
	\} \cap \REC(\wei,\capa).
\end{align*}

% subsection definitions (end)

\subsection{Solving Average-bounded Recharge Games} % (fold)
\label{sub:complexity}

First, we show that solving average-bounded recharge games for a given threshold~$t$ and a given recharge capacity~$\capa$ is $\EXPTIME$-complete and that the problem is still $\EXPTIME$-hard, if the capacity is existentially quantified and only the threshold is given. Formally, we are interested in the following problems:
\begin{problem}
\label{prob_solving_recharge}
Solving Average-bounded recharge games\\ 
\textbf{Input}: Arena $\A = (V, V_0, V_1, E, v_{\initial})$, 
$ \wei \colon E \to -\mathbb{N} \cup \{\res \} $, $\capa \in \mathbb{N}$, and $t \in \mathbb{N}$. \\
\textbf{Question}: Does Player~$0$ win $(\A, \AREC(\wei, \capa, \thres))$?
\end{problem}
\medskip

\begin{problem}
\label{prob_solving_recharge_existscap}
Solving Average-bounded recharge games with existentially quantified capacity\\ 
\textbf{Input}: Arena $\A = (V, V_0, V_1, E, v_{\initial})$, 
$ \wei \colon E \to -\mathbb{N} \cup \{\res \} $, and $t \in \mathbb{N}$. \\
\textbf{Question}: Exists $ \capa \in \mathbb{N}$ s.t.\ Player~$0$ wins $(\A, \AREC(\wei, \capa, \thres))$?
\end{problem}

First, we consider Problem~\ref{prob_solving_recharge}.

\begin{theorem}
\label{thm_rechargeexptime}
	Solving average-bounded recharge games is $\EXPTIME$-complete.
\end{theorem}

We begin the proof by presenting an exponential time algorithm for solving average-bounded recharge games by reducing them to mean-payoff games, similarly to the reduction from lower- and upper-bounded energy games to mean-payoff games~\cite{AvENgergy15}. The mean-payoff objective is given by \[\MP(\wei,\thres) = \{ v_0 v_1 v_2 \cdots \in V^\omega \mid \limsup_{n \to \infty} \frac{1}{n}\EL(v_0 \cdots v_{n-1}) \le t \} .\]

\begin{lemma}
	Average-bounded recharge games can be solved in exponential time.
\end{lemma}

\begin{proof} 
Fix an arena $\A = (V, V_0, V_1, E, v_{\initial})$, 
$ \wei \colon E \to -\mathbb{N} \cup \{\res \} $, $\capa \in \mathbb{N}$, and $t \in \mathbb{N}$. We construct a memory structure~$\mem = (M, m_{\initial}, \update)$ to reduce the average-bounded recharge game to a mean-payoff game.
To this end, let $M = \{0, \ldots, \capa\} \cup \{\bot\}$, $m_{\initial} = \capa$, $\update(\bot, (v,v')) = \bot$, and
\[
\update(c,(v,v'))=
\begin{cases}
\capa & \text{if } \wei(v,v') = \res,\\
c + \wei(v,v') & \text{if } c + \wei(v,v') \ge 0,\\
\bot & \text{if } c + \wei(v,v') < 0.
\end{cases}
\]
Intuitively, the memory structure keeps track of the energy level as long as it is non-negative. If it is negative, then a sink state is reached. Finally, we define a new weight function $\wei'$ by $\wei'((v,c),(v',m)) = c$ for every $c \in M \setminus \{\bot\}$ and $m \in M$ and $\wei'((v,\bot),(v',\bot)) = t+1$. 

\begin{remark}
Let $\play = v_0 v_1 v_2 \cdots$ and $\play' = (v_0, m_0) (v_1, m_1) (v_2, m_2) \cdots$ be such that $\play$ is a play in $\A$ and $\play'$ is the corresponding extended play in $\A \times \mem$.
\begin{enumerate}
	\item If there is no $s \le n$ such that $\EL_\capa(v_0 \cdots v_s) < 0$, then $m_n = \EL_\capa(v_0 \cdots v_n)$.
	\item If there is an $s \le n$ such that $\EL_\capa(v_0 \cdots v_s) < 0$, then $m_n = \bot$.
		
	\item If there is no $s$ such that $\EL_\capa(v_0 \cdots v_s) < 0$, then 
	\[\limsup_{n \to \infty} \frac{1}{n}\EL((v_0,m_0) \cdots (v_{n-1}, m_{n-1})) = \limsup_{n \to \infty} \frac{1}{n} \sum\nolimits_{i = 0}^{n-1} 
	\EL_{\capa}(v_0  \cdots v_i) .\]
		\item If there is an $s$ such that $\EL_\capa(v_0 \cdots v_s) < 0$, then \[\limsup_{n \to \infty} \frac{1}{n}\EL((v_0,m_0) \cdots (v_{n-1}, m_{n-1}) = t+1.\]
		\item $\play \in \AREC(\wei, \capa, \thres)$ if, and only if, $\play' \in \MP(\wei', t)$.
\end{enumerate} 	
\end{remark}

Thus, we have $(\A, \AREC(\wei, \capa, \thres)) \le_\mem (\A\times\mem, \MP(\wei', t))$. Hence, positional determinacy of mean-payoff games~\cite{EhrenfeuchtM79}, Lemma~\ref{lemma_reductionlemma}, and mean-payoff games being solvable in pseudo-polynomial time~\cite{Zwick96} yield the exponential time algorithm.
\end{proof}

An application of Lemma~\ref{lemma_reductionlemma} additionally yields an upper bound on the necessary memory states to implement a winning strategy.

\begin{corollary}
If Player~$0$ wins an average-bounded recharge game with capacity~$\capa$, then she also wins it with a finite-state strategy of size $\capa+2$.
\label{col:memory}
\end{corollary}

\begin{wrapfigure}{r}{.33\textwidth}
\vspace{-19pt}
\begin{center}
	\begin{tikzpicture}[auto,node distance=2 cm, yscale = 1,
		transform shape, thick]

     \tikzstyle{p0}=[draw,circle,text centered,minimum size=7mm,text width=4mm]
     \tikzstyle{p1}=[draw,rectangle,text centered,minimum size=7mm,text width=4mm]

	  \node[p1]	(q1)   				  {$v_0$};
	  \node[p0]	(q0)  	[right of=q1]		{$v_1$};

	  \path[->,>=latex]
	  (-0.9,0) edge [] node [] {} (q1)
		(q0)  edge [bend left] node [] {$R$} (q1) 
		(q1)  edge [bend left] node [] {$-1$} (q0)

		;
	
		\path[->,every loop/.style={looseness=9},>=latex] 
		%	(q1) edge  [in=350,out=10,loop] node[pos=0.5] {$\#$} ()
		
			(q0) edge  [in=20,out=-20,loop] 
			node[pos=0.5,swap]	{$-1$} ()
		
		%	(q1) edge  [in=175,out=215,loop] 
		%	node[pos=0.5]	{$R$} ()
			 ;

	  \end{tikzpicture}
\end{center}	
\caption{The arena for the lower bound on memory requirements in average-bounded recharge games.}

\label{fig_memlowerbound}
\vspace{-5pt}
\end{wrapfigure}
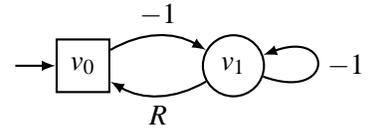
Conversely, it is straightforward to show that this bound is tight: consider the average-bounded recharge game depicted in Figure~\ref{fig_memlowerbound} with some fixed even capacity~$\capa$ and threshold~$\thres = \frac{\capa}{2}$. With $\capa$ memory states, Player~$0$ can implement a strategy whose unique consistent play has the form $(v_0 v_1^{\capa})^\omega$ which has the energy levels~$(\capa, \capa-1, \ldots, 1, 0)^\omega$, which results in a long-run average of $\thres$. However, with $n < \capa$ memory states, the best Player~$0$ is able to do is to implement a strategy whose unique consistent play has the form $(v_0 v_1^{n})^\omega$ which has the energy levels~$(\capa, \capa-1, \ldots, \capa- n)^\omega$, which results in a long-run average of $(\capa-n) + \frac{n}{2} = \capa - \frac{n}{2} > \capa - \frac{\capa}{2} = \thres$. Every other play that is implementable with $n$ memory states has an even higher average. Thus, Player~$0$ needs $\capa$ memory states to meet the bound on the average.

Next, we give an $\EXPTIME$ lower bound by a reduction from countdown games. The arena~$\A = (V, V_0, V_1, E, v_{\initial})$ and the weight function~$\wei$ of such a game are subject to some restrictions: 
\begin{enumerate}
	\item The initial vertex is in $V_0$ and there is a designated sink vertex~$v_\bot \in V_1$ with a self loop,
	\item every vertex in $V_0$ has an edge to $v_\bot$ and all other edges are in $ V_0 \times (V_1 \setminus \{v_\bot\}) \cup (V_1 \setminus \{v_\bot\}) \times V_0 $, 
	\item all edges in $V_0 \times (V_1 \setminus \{v_\bot\})$ have negative weight and there are no two outgoing transitions from a vertex in $V_0$ with the same weight, and
	\item all other edges have weight zero.
\end{enumerate}
The objective is given as
\[
\CNTDWN(\wei, c) = \{ v_0 v_1 v_2 \cdots \in V^\omega \mid \exists n.\, v_n = v_\bot \text{ and } c + EL(v_0 \cdots v_n) = 0 \}.
\]
Intuitively, Player~$0$ picks negative weights that are subtracted from the initial energy~$c$ and Player~$1$ picks the next vertex to continue at (vertices of the countdown game are in $V_0$, $V_1$ only contains auxiliary vertices). Player~$0$ wins if the energy level is exactly zero at some point, at which she has to move to the sink vertex. Otherwise, Player~$1$ wins. Solving countdown games is $\EXPTIME$-complete~\cite{JurdzinskiSL08}. Our reduction is a straightforward adaption of the reduction from countdown to average-energy games~\cite{AvENgergy15}.

\begin{lemma}
\label{lemma_solvingavboundedrecharge}
	Solving average-bounded recharge games is $\EXPTIME$-hard.
\end{lemma}

\begin{proof}
Fix $\A = (V, V_0, V_1, E, v_{\initial})$ and $\wei$ satisfying the requirements of a countdown game and some initial energy $c$. We add a fresh vertex $v_{\initial}'$ to $V_1$, add an edge from $v_{\initial}'$ to $v_{\initial}$ and label it with the recharge action~$\res$ to obtain the arena~$\A'$ and the weight function~$\wei'$. As every play that does not reach the sink vertex traverses infinitely many edges with negative weight, we have $\play \in \CNTDWN(\wei, c)$ if, and only if, $v_{\initial}' \cdot \play \in \AREC(\wei', c, 0)$. Thus, Player~$0$ wins $(\A', \AREC( \wei',c, 0))$ if, and only if, she wins $(\A, \CNTDWN(\wei, c))$. Hence, solving average-bounded recharge games is $\EXPTIME$-hard.
\end{proof}

Note that the hardness depends on the requirement to bound the average. Recharge games without average-bound are solvable in pseudo-polynomial time, as such a game can be expressed as a one-dimensional consumption game~\cite{DBLP:conf/cav/BrazdilCKN12}. Determining the minimal cover (the analogue of our capacity in consumption games, see \cite{DBLP:conf/cav/BrazdilCKN12} for a formal definition) for the initial vertex and comparing it to the given capacity yields the desired result, as the minimal cover in a one-dimensional consumption game can be computed in pseudo-polynomial-time~\cite{DBLP:conf/cav/BrazdilCKN12}. Whether recharge games can be solved in polynomial time is open. In the next subsection, we present a variant that is solvable in polynomial time.

Also, the previous hardness proof can be adapted to recharge games with a given threshold and existentially quantified capacity (Problem~\ref{prob_solving_recharge_existscap}). To this end, we add the initial gadget presented in Figure~\ref{fig_gadget} to a countdown game~$\game$. In order to win this game, Player~$0$ has to reach the Player~$1$ vertex with energy level $c$. If the energy level is larger then Player~$1$ can take the edge with weight~$-c$ and reach the sink with a positive energy level. Hence, the average accumulated energy will be non-zero, too. Conversely, if the energy level is smaller than $c$, then taking the same edge yields a negative energy level. Hence, in both cases the objective~$\AREC(\wei, \capa, 0)$ is violated, independently of the value of $c$. However, if Player~$0$ reaches the Player~$1$ vertex with energy level~$c$, then she wins from there, if and only if, she has a winning strategy for the countdown game~$\game$ with initial value~$c$. Thus, she wins the recharge game with objective~$\AREC(\wei, \capa, 0)$ for some $\capa$ if, and only if, she wins the countdown game~$\game$ with objective~$\CNTDWN(\wei, c)$.

\begin{figure}[h]
\begin{center}
	\begin{tikzpicture}[auto,node distance=2 cm, yscale = 1,
		transform shape, thick]

     \tikzstyle{p0}=[draw,circle,text centered,minimum size=7mm,text width=4mm]
     \tikzstyle{p1}=[draw,rectangle,text centered,minimum size=7mm,text width=4mm]

	  \node[p0]	(q1)   				  {};
	  \node[p1]	(q0)  	[right of=q1] {};
	  \node[p0] (q2)    [right of=q0] {};

	  \path[->,>=latex]
	  (-0.9,0) edge [] node [] {} (q1)
		(q1)  edge [bend right = 30] node [below] {$0$} (q0) 
		(q0)  edge node [] {$-c$} (q2) 
		(q0)  edge [bend right = 20] node [below] {$0$} (6.5,0)

		;
	
		\path[->,every loop/.style={looseness=9},>=latex] 
		%	(q1) edge  [in=350,out=10,loop] node[pos=0.5] {$\#$} ()
		
			(q1) edge  [loop right] 
			node[pos=0.5,swap, right]	{~$-1$} ()
			
			(q2) edge  [loop right] 
			node[pos=0.5,swap, right]	{~$0$} ()
		
		%	(q1) edge  [in=175,out=215,loop] 
		%	node[pos=0.5]	{$R$} ()
			 ;
	
		\draw[rounded corners] (6,-.7) rectangle  (10,.7);
		\node at (8, 0) {$\game$};

	  \end{tikzpicture}
\end{center}		
\caption{The gadget for showing Problem~\ref{prob_solving_recharge_existscap} $\EXPTIME$-hard.}
\label{fig_gadget}
\end{figure}
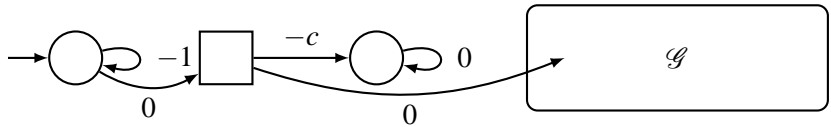

\begin{theorem}
Solving average-bounded recharge games with existentially quantified capacity and a given threshold is $\EXPTIME$-hard.
\end{theorem}

However, it is an open problem whether these games can be solved in exponential time. The reduction to mean-payoff games presented above depends on the capacity being part of the input. This is related to the absence of good upper bounds on the necessary capacity to achieve a given threshold. 

% subsection complexity (end)

\subsection{Finding a Sufficient Capacity in Recharge Games} % (fold)
\label{sub:existence_of_cap}

To tackle the high complexity of solving average-bounded recharge games, we consider the problem where the recharge capacity~$\capa$ and the threshold~$\thres$ are existentially quantified. As the energy level is always bounded from above by $\capa$, which implies that the average accumulated energy is also bounded by $\capa$, it suffices to consider the objective~$\REC(\wei, \capa)$, analogous results hold for the objective~$\AREC(\wei, \capa, \thres)$. We show that the following problem can be solved in polynomial time. 

\begin{problem}
Existence of a sufficient recharge level in recharge games\\ 
\textbf{Input}: Arena $\A = (V, V_0, V_1, E, v_{\initial})$ and 
$ \wei \colon E \to -\mathbb{N} \cup \{\res\} $ \\
\textbf{Question}: Exists a capacity~$\capa$ s.t.\ Player~$0$ wins $(\A, \REC(\wei, \capa))$?
\end{problem}

One attempt to prove this result is to again encode the game as a one-dimensional consumption game as described above. However, this only yields a pseudo-polynomial time algorithm. In the following, we present a truly polynomial time algorithm by a reduction to three-color parity games. Given a coloring~$\Omega \colon V \to \mathbb{N}$,  $\PARITY(\Omega)$ denotes the (max)-parity objective, which contains all plays~$v_0 v_1 v_2 \cdots \in V^\omega$ such that the maximal color appearing infinitely often in $\Omega(v_0) \Omega(v_1) \Omega(v_2) \cdots $ is even. 

\begin{theorem}
	The existence of a sufficient recharge level in a recharge game can be determined in polynomial time.
\end{theorem}

\begin{proof}
Fix an arena $\A = (V, V_0, V_1, E, v_{\initial})$ and 
$ \wei \colon E \to -\mathbb{N} \cup \{\res\} $. We construct a three-color parity game with the following property: Player~$0$ wins the parity game if, and only if, there is a $\capa$ such that Player~$0$ wins $(\A, \REC(\wei, \capa))$. We assume w.l.o.g.\ that every vertex of $\A$~either only has incoming edges labeled with $\res$, only has incoming edges labeled with $0$, or
 only has incoming edges labeled with a negative weight.
This can always be achieved by tripling the set of vertices, one copy for each type of incoming edge. The new initial vertex is some fixed copy of the original initial vertex. This transformation does not change the winner and only results in a linear increase in the number of states.

Now, we can speak of recharge-vertices, zero-vertices, and of decrement-vertices and define the coloring $\Omega$ such that it assigns color~$2$ to the recharge-vertices, color~$1$ to the decrement-vertices, and color~$0$ to the zero-vertices. We claim that Player~$0$ has a winning strategy for the induced parity game if, and only if, there is a $\capa$ such that Player~$0$ wins $(\A, \REC(\wei, \capa))$. 

First, assume Player~$0$ has a winning strategy for the parity game, which we can assume w.l.o.g.\ to be positional~\cite{EmersonJutla91,Mostowski91}. Let $W$ be the largest absolute weight in the image of $\wei$ and define $\capa = (|V|-1)\cdot W$. We claim that $\sigma$ is a winning strategy for Player~$0$ in $(\A, \REC(\wei, \capa))$. Assume it is not: then, there is a play prefix~$v_0 \cdots v_n$ that is consistent with $\sigma$ such that $\EL_\capa(v_0 \cdots v_n) <0$. Let $v_i \cdots v_n$ be the suffix since the last recharge edge was traversed, i.e., $-\EL(v_i \cdots v_n) > \capa$. By the choice of $\capa$, there are positions $j$ and $j'$ satisfying $i < j < j' \le n$ such that $v_j = v_{j'}$ and $\EL(v_j \cdots v_{j'}) <0$, i.e., there is a cycle with negative cost and without recharge edge. As $\sigma$ is positional, the play~$v_0 \cdots v_{j-1}(v_j \cdots v_{j'-1})^\omega$ obtained by reaching and then repeating this cycle is consistent with $\sigma$ as well. However, in the parity game, this cycle visits no recharge-vertex, but at least one decrement-vertex. Hence, it is losing for Player~$0$, which contradicts $\sigma$ being a winning strategy. Hence, $\sigma$ is indeed also a winning strategy for $(\A, \REC(\wei, \capa))$.

Now, assume there is some $\capa$ and a strategy~$\sigma$ that is winning for Player~$0$ in $(\A, \REC(\wei, \capa))$. We claim that this strategy is also winning for her in the parity game. Assume, it is not, i.e., there is a play that is consistent with $\sigma$, but losing for Player~$0$ in the parity game. By our choice of colors, this implies that this play visits only finitely many recharge-vertices, but infinitely many decrement-vertices. Thus, it has a prefix whose recharge energy level is negative. But this contradicts the fact that $\sigma$ is a winning strategy for the recharge game. 

To conclude, it remains to remark that three-color parity games can be solved in polynomial time. 
\end{proof}

By applying both directions of the equivalence, we obtain the following corollary.

\begin{corollary}
If there is a $\capa$ such that Player~$0$ wins $(\A, \REC(\wei, \capa))$, then she also wins $(\A, \REC(3\cdot(n-1)\cdot W, \wei))$, where $n$ is the number of vertices of $\A$ and $W$ is the largest absolute weight in the domain of $\wei$. Player~$0$ wins the latter game with a finite-state strategy of size three.
\end{corollary}

Note that this can be improved slightly by a finer analysis: the factor~$(n-1)$ can be replaced by the number of decrement-vertices. Conversely, it is straightforward to construct examples that prove these bounds to be tight, e.g., a cycle of $n$ edges, one being a recharge edge and all others having weight~$-W$.
% This game is won by Player~$0$, provided the recharge capacity is at least $(n-1)\cdot W$.

% subsection existence_of_cap (end)

% section recharge_games (end)

%% file: tradeoffs.tex
%!TEX root = main.tex

\section{Tradeoffs in Recharge Games} % (fold)
\label{sec:tradeoffs}

In this section, we illustrate two different tradeoff scenarios between different quality measures for winning strategies that occur in average-bounded recharge games, i.e., tradeoffs between capacity and long-run average and between memory size and long-run average. Note that increasing the recharge capacity in such a game has a (possibly negative) influence on the long-run average, as every recharge returns the energy level to the capacity. All games we consider here are solitaire games for Player~$0$, i.e., every vertex belongs to Player~$0$. Thus, a strategy can be identified with the unique play consistent with it.

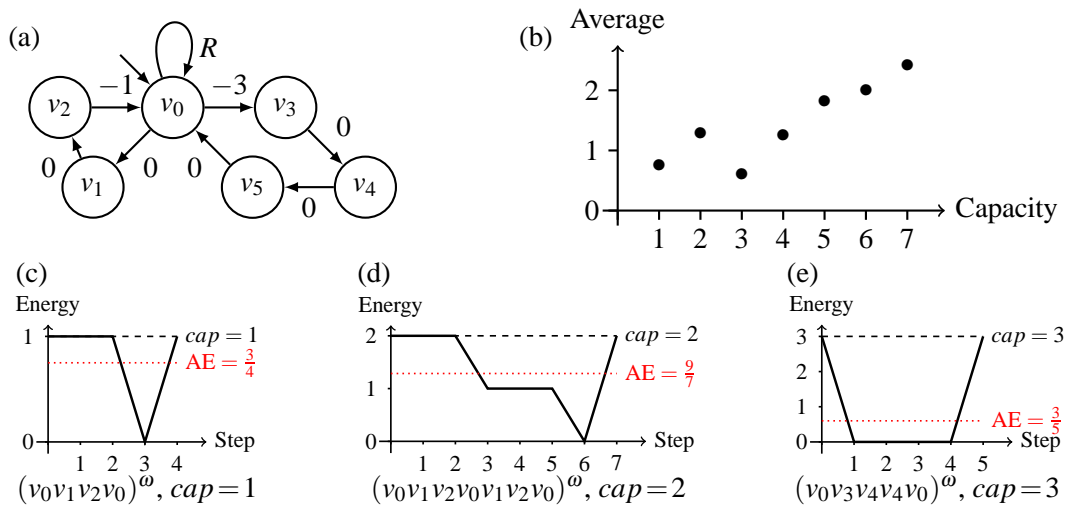
\begin{figure}[htb]
	\centering

\begin{tikzpicture}
	
	\clip (-8,-5.1) rectangle (8,1.6);
	
	\def \yplots {-3.4}
	\def \ylabels {-1.95}
	\def \ycaps {-4.8}

	\node at (-6.8, 1.2) {(a)};
	\node at (0.0, 1.2) {(b)};
	\node at (-6.7, \ylabels) {(c)};
	\node at (-2.15, \ylabels) {(d)};
	\node at (3.55, \ylabels) {(e)};

%%%%%%%%%%%%%%%%% arena

		\node at (-4.3,-.2) {
		\begin{tikzpicture}[auto,node distance=1.5 cm, yscale = 1,
		transform shape, thick]

     \tikzstyle{p0}=[draw,circle,text centered,minimum size=7mm,text width=4mm]
     \tikzstyle{p1}=[draw,rectangle,text centered,minimum size=7mm,text width=4mm]
			
	  \node[p0]	(q1)                {$v_0$};
	  \node[p0]	(q2)  [right of=q1]	{$v_3$};
	  \node[p0]	(q3)  [below right of=q2]	{$v_4$};
	  \node[p0]	(q4)  [ left of=q3]	{$v_5$};
	  
	  \node[p0]	(q6)  [left of=q1]	{$v_2$};
	 % \node[circ]	(q6)  [below left of=q5]	{};
	  \node[p0]	(q5)  [below left of=q1]	{$v_1$};

	  \path[->,>=latex]
	  (-0.7,0.7) edge [] node [] {} (q1)
		(q1)  edge [] node [] {$-3$} (q2) 
		(q2)  edge [] node [] {$0$} (q3) 
		(q3)  edge [] node [] {$0$} (q4) 
		(q4)  edge [] node [] {$0$} (q1) 
		
		(q1)  edge [] node [] {$0$} (q5) 
		(q5)  edge [] node [] {$0$\,\,} (q6) 
		(q6)  edge [] node [] {$-1$} (q1)

		;
	
		\path[->,every loop/.style={looseness=9},>=latex] 
		%	(q1) edge  [in=350,out=10,loop] node[pos=0.5] {$\#$} ()
		
			(q1) edge  [in=70,out=110,loop] 
			node[pos=0.83]	{$R$} () ;

		\node at (0,-2.2) {  };

	\end{tikzpicture} 
		};
		
%%%%%%%%%%%%%%%%% plot
		
		\node at (3.7,0) {
		\begin{tikzpicture}[thick]

  		\def \xscale {0.55}
  		\def \yscale {0.8}
  		\def \xmax {7}
  		\def \ymax {2.5}
	
  		\def\xy#1#2{({#1*\xscale},{#2*\yscale})}
			   
  	    \draw[->] (-0.2,0) -- ({\xmax*\xscale+0.5},0) node[right] {Capacity};
  	    \draw[->] (0,-0.2) -- (0,{\ymax*\yscale + 0.2}) node[above] {Average};
 		
          \foreach \y in {0,1,...,\ymax} {	
  			\draw (-0.1,{\y*\yscale}) node[anchor=east] {\y} ;
  		}
		
          \foreach \x in {1,2,...,\xmax} {
  			\draw ({\x*\xscale},-0.1) node[anchor=north] {\x} ;
  		}

          \foreach \x in {0,...,\xmax} {
              \draw ({\x*\xscale},0) -- ({\x*\xscale},-3.5pt);
          }

          \foreach \y in {0,...,\ymax} {
              \draw (0, {\y*\yscale}) -- (-3.5pt, {\y*\yscale});
          }
		
		\foreach \Point in {\xy{1}{3/4}, \xy{2}{9/7}, \xy{3}{3/5}, \xy{4}{5/4}, \xy{5}{20/11}, \xy{6}{2}, \xy{7}{29/12}}
		{
		    \node at \Point {\textbullet};
		}
  	\end{tikzpicture}

		};

%%%%%%%%%%%%%%%%% evolution cap 1

		\node at (-5.3,\ycaps) {$({v_0 v_1 v_2 v_0})^\omega$, $\capa \! = \!1$};
		\node at (-5.3,\yplots) {\scalebox{.78}{
		\begin{tikzpicture}[thick]

		\def \xscale {0.55}
		\def \yscale {1.8}
		\def \xmax {4}
		\def \ymax {1}
		
		\def\xy#1#2{({#1*\xscale},{#2*\yscale})}
			   
	    \draw[->] (-0.2,0) -- ({\xmax*\xscale+0.5},0) node[right] {Step};
	    \draw[->] (0,-0.2) -- (0,{\ymax*\yscale + 0.2}) node[above] {Energy};
 		
        \foreach \y in {0,1,...,\ymax} {	
			\draw (-0.1,{\y*\yscale}) node[anchor=east] {\y} ;
		}
		
        \foreach \x in {1,2,...,\xmax} {
			\draw ({\x*\xscale},-0.1) node[anchor=north] {\x} ;
		}

        \foreach \x in {0,...,\xmax} {
            \draw ({\x*\xscale},0) -- ({\x*\xscale},-1.5pt);
        }

        \foreach \y in {0,...,\ymax} {
            \draw (0, {\y*\yscale}) -- (-1.5pt, {\y*\yscale});
        }
		
		\def\caplevel{1}
		\draw[thick,dashed,color=black]
			\xy{0}{\caplevel} -- \xy{\xmax}{\caplevel};
		
		\draw \xy{\xmax}{\caplevel} node[anchor=west] {$ {cap} = 1$ };

		\draw[very thick,-,color=black]
					\xy{0}{1} --
					\xy{1}{1} --
					\xy{2}{1} --
					\xy{3}{0} --
					\xy{4}{1}  
					 ;
					 
		\def\AElevel{3/4}
		\draw[dotted,color=red,thick]
					\xy{0}{\AElevel} -- \xy{\xmax}{\AElevel};
		
		\draw \xy{\xmax}{\AElevel} node[anchor=west,red] 
		{$ \mathrm{AE} = \frac{3}{4}$ };
				
	\end{tikzpicture}
		}};

%%%%%%%%%%%%%%%%% evolution cap 2

		\node at (-0.1,\ycaps) {$({v_0 v_1 v_2 v_0 v_1 v_2 v_0})^{\omega}$, $\capa \! = \!2$};		
		\node at (-0.1,\yplots) {\scalebox{.78}{
		\begin{tikzpicture}[thick]

		\def \xscale {0.55}
		\def \yscale {0.9}
		\def \xmax {7}
		\def \ymax {2.2}
			
		\def\xy#1#2{({#1*\xscale},{#2*\yscale})}
			   
	    \draw[->] (-0.2,0) -- ({\xmax*\xscale+0.5},0) node[right] {Step};
	    \draw[->] (0,-0.2) -- (0,{\ymax*\yscale }) node[above] {Energy};
 		
        \foreach \y in {0,1,...,\ymax} {	
			\draw (-0.1,{\y*\yscale}) node[anchor=east] {\y} ;
		}
		
        \foreach \x in {1,2,...,\xmax} {
			\draw ({\x*\xscale},-0.1) node[anchor=north] {\x} ;
		}

        \foreach \x in {0,...,\xmax} {
            \draw ({\x*\xscale},0) -- ({\x*\xscale},-1.5pt);
        }

        \foreach \y in {0,...,\ymax} {
            \draw (0, {\y*\yscale}) -- (-1.5pt, {\y*\yscale});
        }
		
		\def\caplevel{2}
		\draw[thick,dashed,color=black]
			\xy{0}{\caplevel} -- \xy{\xmax}{\caplevel};
		
		\draw \xy{\xmax}{\caplevel} node[anchor=west] {$ {cap} = 2$ };

		\draw[very thick,-,color=black]
					\xy{0}{2} --
					\xy{1}{2} --
					\xy{2}{2} --
					\xy{3}{1} --
					\xy{4}{1} --
					\xy{5}{1} --
					\xy{6}{0} -- 
					\xy{7}{2};
		
		\def\AElevel{9/7}
		\draw[dotted,color=red,thick]
					\xy{0}{\AElevel} -- \xy{\xmax}{\AElevel};
		
		\draw \xy{\xmax}{\AElevel} node[anchor=west,red] 
		{$ \mathrm{AE} = \frac{9}{7}$ };

	\end{tikzpicture}
	}};

%%%%%%%%%%%%%%%%% evolution cap 3

		\node at (5.2,\ycaps) {$({v_0 v_3 v_4 v_4 v_0})^\omega$, $\capa \! =\! 3$};
		\node at (5.2,\yplots) {\scalebox{.78}{
			\begin{tikzpicture}[thick]

		\def \xscale {0.55}
		\def \yscale {0.6}
		\def \xmax {5}
		\def \ymax {3}
			
		\def\xy#1#2{({#1*\xscale},{#2*\yscale})}
			   
	    \draw[->] (-0.2,0) -- ({\xmax*\xscale+0.5},0) node[right] {Step};
	    \draw[->] (0,-0.2) -- (0,{\ymax*\yscale + 0.2}) node[above] {Energy};
 		
        \foreach \y in {0,1,...,\ymax} {	
			\draw (-0.1,{\y*\yscale}) node[anchor=east] {\y} ;
		}
		
        \foreach \x in {1,2,...,\xmax} {
			\draw ({\x*\xscale},-0.1) node[anchor=north] {\x} ;
		}

        \foreach \x in {0,...,\xmax} {
            \draw ({\x*\xscale},0) -- ({\x*\xscale},-1.5pt);
        }

        \foreach \y in {0,...,\ymax} {
            \draw (0, {\y*\yscale}) -- (-1.5pt, {\y*\yscale});
        }
		
		\def\caplevel{3}
		\draw[thick,dashed,color=black]
			\xy{0}{\caplevel} -- \xy{\xmax}{\caplevel};
		
		\draw \xy{\xmax}{\caplevel} node[anchor=west] {$ \capa = 3$ };

		\draw[very thick,-,color=black]
					\xy{0}{3} --
					\xy{1}{0} --
					\xy{2}{0} --
					\xy{3}{0} --
					\xy{4}{0} --
					\xy{5}{3}  ;
		
		\def\AElevel{3/5}
		\draw[dotted,color=red,thick]
					\xy{0}{\AElevel} -- \xy{\xmax}{\AElevel};
		
		\draw \xy{\xmax}{\AElevel} node[anchor=west,red] 
		{$ \mathrm{AE} = \frac{3}{5}$ };

	\end{tikzpicture}

		}};
			
	\end{tikzpicture}

	\caption{(a) An average-bounded recharge game with tradeoff between capacity and long-run average. (b) A plot of the tradeoff. (c) - (e) Energy progressions of different plays in the average-bounded
	recharge game for different capacities.}
	\label{fig:tradeoffeks1}

\end{figure}

First, we study the tradeoff between the capacity and the long-run
average energy level. Consider the game in Figure~\ref{fig:tradeoffeks1}(a): 
Player~$0$ wins the game for $\capa = 1$ and $\thres = 1$ by realizing
the long-run average $\frac{3}{4}$ with the play $({v_0
v_1 v_2 v_0})^\omega$ (Figure~\ref{fig:tradeoffeks1}(c)). But, by increasing the capacity
 to $\capa=2$, it is no longer possible for her to win for $\thres=1$, as
the best long-run average she can realize is $\frac{9}{7}$ by playing
$({v_0 v_1 v_2 v_0 v_1 v_2 v_0})^{\omega}$ (Figure~\ref{fig:tradeoffeks1}(d)).
However, for $\capa=3$, she can again win for $\thres = 1$, 
and it is possible to realize the long-run average $\frac{3}{5}$ by playing
$({v_0 v_3 v_4 v_5 v_0})^\omega$ (Figure~\ref{fig:tradeoffeks1}(e)). Again, with 
$\capa=4$ Player~$0$ loses for $\thres = 1$. 

This example shows that higher capacity can be traded for a lower long-run average and that the tradeoff is non-monotonic. Figure~\ref{fig:tradeoffeks1}(b) shows a plot of the
tradeoff for capacities ranging from  $1$ to $7$.

\begin{wrapfigure}{R}{.37\textwidth}
\vspace{-10pt}

	\centering
	\begin{tikzpicture}[thick]

		\def \xscale {0.5}
		\def \yscale {0.5}
		\def \xmax {5}
		\def \ymax {4}
			
		\def\xy#1#2{({#1*\xscale},{#2*\yscale})}
			   
	    \draw[->] (-0.2,0) -- ({\xmax*\xscale+0.5},0) node[right] {Memory};
	    \draw[->] (0,-0.12) -- (0,{\ymax*\yscale + 0.2}) node[above] {Average};
 		
		%         \foreach \y in {0,1,...,\ymax} {
		% 	\draw (-0.1,{\y*\yscale}) node[anchor=east] {\y} ;
		% }
		%
		%         \foreach \x in {1,2,...,\xmax} {
		% 	\draw ({\x*\xscale},-0.1) node[anchor=north] {\x} ;
		% }

        \foreach \x in {0,...,\xmax} {
            \draw ({\x*\xscale},0) -- ({\x*\xscale},-3.5pt);
        }

        \foreach \y in {0,...,\ymax} {
            \draw (0, {\y*\yscale}) -- (-3.5pt, {\y*\yscale});
        }

		\def\caplevel{4}
		\draw \xy{-0.1}{\caplevel} node[anchor=east] {$\capa-\frac{1}{2}$};
		\draw \xy{-0.1}{\caplevel/2} node[anchor=east] {$\frac{\capa}{2}$};
		\draw \xy{-0.2}{0} node[anchor=east] {$0$};
		
		 \draw \xy{\caplevel}{-0.15} node[anchor=north] {$\capa$};
		 \draw \xy{0}{-0.15} node[anchor=north] {$1$};

		\draw[very thick,-,color=black]
					\xy{0}{\caplevel} --
					\xy{\caplevel}{\caplevel/2} --
					\xy{6}{\caplevel/2}  ;
					 
	\end{tikzpicture}
	\caption{A plot of the tradeoff between memory size and long-run average in the game in Figure~\ref{fig_memlowerbound}.}
	\label{fig:memtradeoff}
	\vspace{-34pt}

\end{wrapfigure}
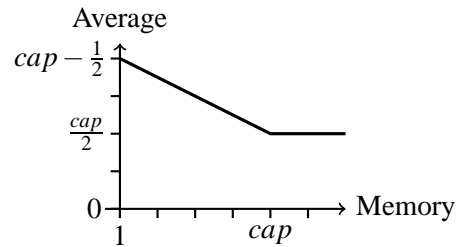

Another tradeoff scenario is between the number of memory states required to implement a
strategy and the long-run average energy level it realizes. Consider the
recharge game from Figure~\ref{fig_memlowerbound}: as discussed below Corollary~\ref{col:memory}, Player~$0$ can win for the threshold $t= \frac{\capa}{2}$ with $\capa$ memory states.
However, with $n < \capa$
memory states, she can only guarantee the long-run average~$(\capa -n ) + \frac{n}{2}$. In particular,  the best long-run average that is realizable by a positional strategy (which requires one memory state to implement) is $\capa - \frac{1}{2}$ (see Fig.~\ref{fig:memtradeoff}).

% subsection memory_vs_average (end) 
 
% section tradeoffs (end)

%% file: conc.tex
\section{Conclusion} % (fold)
\label{sec:conc}

We continued the study of average-energy games by considering problems where the bound on the average is existentially quantified instead of given as part of the input. We showed that solving this problem is equivalent to determining whether the maximal energy level can be uniformly bounded by a strategy. The latter problem is known to be decidable in doubly-exponential time, which therefore also holds for our original problem. Then, we considered a different type of energy evolution where energy is only consumed or reset to some fixed capacity. Solving the average-bounded variants of these games is shown to be complete for exponential time. Due to this high complexity, we again considered a variant where the bounds are existentially quantified. This problem turns out to be solvable in polynomial time. Finally, we studied tradeoffs between the different bounds and the memory requirements of winning strategies: increasing the upper bound on the maximal energy level is shown to allow to improve the average energy level and memory can be traded for smaller upper bounds and vice versa.

For future work, it would be interesting to extend our results to a multi-dimensional setting. Also, the exact complexity of determining the existence of an upper bound in average-energy games is open. Finally, the decidability of average-energy games with a given threshold, but without an upper bound on the energy level is open~\cite{AvENgergy15}. In current work, we study whether our approach presented in Section~\ref{sec:existential_average_energy_games} can be adapted to solve these problems, e.g., by not picking representatives by minimizing peak height but some other measure. These questions are also related to the complexity of recharge games with a given threshold where the capacity is existentially quantified. Finally, we are studying upper bounds on the tradeoffs presented in Section~\ref{sec:tradeoffs}.

\paragraph*{Acknowledgements.} The work presented here was carried out while the third author visited the Distributed and Embedded Systems Unit at
Aalborg University and while the second author visited the Reactive Systems Group at Saarland University. We thank these institutions for their hospitality.